\pdfoutput=1

\documentclass[11pt]{article}

\usepackage{appendix}

\usepackage[figuresright]{rotating}
\usepackage[applemac]{inputenc}

\usepackage{tikz}
\usepackage{pgfplots}
\usepgfplotslibrary{groupplots} 
\usetikzlibrary{pgfplots.groupplots}

\usepgfplotslibrary{fillbetween}

\usepackage{graphicx}
\usepackage{dcolumn}
\usepackage{bm}
\usepackage{amscd,amsthm}
\usepackage{ifthen}
\usepackage{booktabs}

\newcommand\Ex{\mathbf E} 

\newcommand{\x}{x^L}
\newcommand{\y}{y^L}

\usepackage{hyperref}
\hypersetup{
    bookmarks=true,         
    unicode=false,          
    pdftoolbar=true,        
    pdfmenubar=true,        
    pdffitwindow=false,     
    pdfstartview={FitH},    
    pdfsubject={Subject},   
    pdfproducer={Producer}, 
    pdfnewwindow=true,      
    colorlinks=true,       
    linkcolor=red,          
    citecolor=green,        
    filecolor=magenta,      
    urlcolor=cyan           
}

\newtheorem{lemma}{Lemma}

\newtheorem{theorem}{Theorem}

\newtheorem{proposition}{Proposition}

\usepackage{fullpage}

\usetikzlibrary{matrix,arrows,decorations.pathmorphing}

\newcommand\figurechannelmodel{
\begin{tikzpicture}[scale=0.7,>=latex]
\draw[gray] (-1,-1.2) rectangle (10.7,2.2);

\begin{scope}[xshift=-3cm]
\node at (-1.2,0) {\small 010110};
\draw[->, out=0,in =-180] [->] (0,0) to node[above,swap] {
\hspace{0.8cm}
\parbox{2cm}{\parbox{1.9cm}{\small encode}}
} (2,0);
\end{scope}

\begin{scope}[xshift=-0.7cm]
\draw[->, out=0,in =-180] [->] (12,0) to node[above,swap] {
\hspace{0.8cm}
\parbox{2cm}{\parbox{1.9cm}{\small decode}}
} (13.8,0);
\node at (14.9,0) {\small 010110};
\end{scope}

\node at (-0,0.4) {\tiny \textcolor{blue}{ACATACGT}};
\node at (-0,0) {\tiny \textcolor{red}{CATGTACA}};
\node at (-0,-0.4) {\tiny \textcolor{brown}{GCTATGCC}};

\draw[->] [->] (1,0) to (1.8,0);
\node at (1.5,0) [above,rotate=60] {\hspace{1.1cm}\scriptsize synthesis};

\begin{scope}[xshift=-10cm,yshift=3cm]
\draw[thick, drop shadow, fill=white] (13,-3) circle (1cm);
\draw[very thick, blue] (13.4,-3) to (13.7,-3);
\draw[very thick, red] (13,-2.45) to (12.8,-2.2);
\draw[very thick, red] (12.2,-3.4) to (12.5,-3.3);
\draw[very thick, brown] (12.7,-3.1) to (13.0,-3);
\draw[very thick, brown] (13,-2.2) to (13.3,-2.3);
\end{scope}

\draw[->] (4.3,0) to node[above,swap,rotate=60] {\hspace{1.5cm}\scriptsize amplification
} (5.1,0);
\begin{scope}[xshift=-6.7cm,yshift=3cm]
\draw[thick, drop shadow, fill=white] (13,-3) circle (1cm);
\draw[very thick, blue] (13.1,-3) to (13.2,-3.3);
\draw[very thick, blue] (13.4,-3) to (13.7,-3);
\draw[very thick, blue] (13.4,-2.7) to (13.5,-2.45);
\draw[very thick, blue] (12.4,-2.7) to (12.5,-2.45);
\draw[very thick, blue] (12.8,-2.7) to (13.1,-2.7);
\draw[very thick, red] (12.8,-3.8) to (13.1,-3.7);
\draw[very thick, red] (12.2,-3.4) to (12.5,-3.3);
\draw[very thick, red] (12.9,-3.4) to (12.65,-3.2);
\draw[very thick, red] (13,-2.45) to (12.8,-2.2);
\draw[very thick, red] (13.5,-3.2) to (13.5,-3.5);
\draw[very thick, brown] (12.5,-3.1) to (12.5,-2.8);
\draw[very thick, brown] (12.7,-3.1) to (13.0,-3);
\draw[very thick, brown] (13.7,-3.2) to (13.9,-3);
\draw[very thick, brown] (13,-2.2) to (13.3,-2.3);
\draw[very thick, brown] (12.4,-3.5) to (12.7,-3.6);
\end{scope}

\draw[->] [->] (7.7,0) to (8.6,0);
\node at (8.3,0.2) [above,rotate=60] {\hspace{1cm}\scriptsize sequencing};
\begin{scope}[xshift=-1.1cm]
\node at (10.7,0.8) {\tiny \textcolor{blue}{ACATA{\bf t}GT}};
\node at (10.7,0.4) {\tiny \textcolor{red}{C{\bf g}TGTACA}};
\node at (10.7,0) {\tiny \textcolor{red}{CATGTACA}};
\node at (10.7,-0.4) {\tiny \textcolor{blue}{ACATACGT{\bf t} }};
\node at (10.7,-0.8) {\tiny \textcolor{blue}{CATACGT}};
\end{scope}

\begin{scope}[yshift=-3.7cm,xshift=-0.9cm,>=latex]
\node [rotate=90] at (5.8,2.1) {$=$};

\draw[gray] (-2,-1.7) rectangle (13.7,1.6);

\node at (1,0.5) {\scriptsize \textcolor{blue}{ACATACGT}};
\node at (1,0) {\scriptsize \textcolor{red}{CATGTACA}};
\node at (1,-0.5) {\scriptsize \textcolor{brown}{GCTATGCC}};

\draw[->] [->] (4,-0.3) to node[above,swap] {sample \& perturb} (8,-0.3);

\draw[decoration={brace,mirror,raise=5pt,amplitude=4pt},decorate]
  (-0.2,-0.5) -- node[below=6pt,yshift=-0.15cm] {$L$} (2.2,-0.5);

\draw [decorate,decoration={brace,amplitude=4pt},xshift=-10pt,yshift=0pt]
(0,-0.7) -- (0,0.7) node [black,midway,xshift=-0.4cm] 
{$M$};

\node at (10.7,1) {\scriptsize \textcolor{blue}{ACATA{\bf t}GT}};
\node at (10.7,0.5) {\scriptsize \textcolor{red}{C{\bf g}TGTACA}};
\node at (10.7,0) {\scriptsize \textcolor{red}{CATGTACA}};
\node at (10.7,-0.5) {\scriptsize \textcolor{blue}{ACATACGT{\bf t} }};
\node at (10.7,-1) {\scriptsize \textcolor{blue}{CATACGT}};

\draw [decorate,decoration={brace,amplitude=4pt,mirror},xshift=1pt,yshift=0pt]
(12,-1.3) -- (12,1.3) node [black,midway,xshift=0.4cm] 
{$N$};

\end{scope}

\end{tikzpicture}
}

\usepackage[
backend=bibtex,
hyperref=auto,
style=alphabetic,
sorting=none,
natbib=true,
firstinits=true,
maxbibnames=10,
]{biblatex}
\bibliography{./bibliography.bib} 

\definecolor{DarkBlue}{rgb}{0,0,0.7}

\usepackage{tikz,pgf,pgfplots}
\usetikzlibrary{calc,shadows}
\usetikzlibrary{pgfplots.groupplots}

\usepackage{times}
\usepackage{xr}

\usepackage{amsthm}
\usepackage{amsmath}    
\usepackage{amssymb}    
\usepackage{mathrsfs}   
\usepackage{epsfig}     
\usepackage{graphicx}
\usepackage{url}

\usepackage{url}
\usepackage{color}
\usepackage{booktabs}

\usepackage{subfigure}
\newcommand\ML{ML} 

\newcommand\NL{NL}

\usepackage{enumitem}

\graphicspath{{}{figs/}{../}{../figs/}{../Monograph/monograph/nowformat/}}

\renewcommand{\caption}[1]{\singlespacing\hangcaption{#1}\normalspacing}

\usepackage{amsmath}
\usepackage{amsfonts}
\usepackage{dsfont}
\usepackage{mathrsfs}

\usepackage{algorithmic}
\usepackage{algorithm}
\usepackage{multirow}

\newcommand{\Z}{{\mathds Z}}

\newcommand{\ep}{\epsilon}

\newcommand{\bA}{\mathsf{A}}
\newcommand{\bC}{\mathsf{C}}
\newcommand{\bG}{\mathsf{G}}
\newcommand{\bT}{\mathsf{T}}

\newcommand{\C}{{\mathcal C}}

\newcommand{\E}{{\mathcal E}}

\newcommand{\M}{{\mathcal M}}
\newcommand{\N}{{\mathcal N}}

\newcommand{\T}{{\mathcal T}}

\newcommand{\q}[2]{Q_{s_{#1},d_{#2}}}

\newtheorem{cor}{Corollary} 

\newcommand{\blfootnote}[1]{%
  \begingroup
  \renewcommand\thefootnote{}\footnote{#1}%
  \addtocounter{footnote}{-1}%
  \endgroup
}

\newcommand{\defi}{\triangleq}

\newcommand{\one}{\mathds 1}

\newcommand{\st}{:}

\newcounter{constcount}
\newcounter{numcount}
\newcommand{\eqnum}{\stackrel{(\roman{numcount})}{=}\stepcounter{numcount}}

\newcommand{\leqnum}{\stackrel{(\roman{numcount})}{\leq\;}\stepcounter{numcount}}

\newcommand{\rescnt}{\setcounter{numcount}{1}}

\newcounter{thmcnt}
  \let\Oldsection\section
  
\renewcommand{\section}{\stepcounter{thmcnt}\Oldsection}

\allowdisplaybreaks

\newcommand{\aln}[1]{\begin{align*}#1\end{align*}}

\newcommand{\al}[1]{\begin{align}#1\end{align}}

\def\Item$#1${\item $\displaystyle#1$
   \hfill\refstepcounter{equation}(\theequation)}

\newcommand{\bea}{\begin{eqnarray}}
\newcommand{\eea}{\end{eqnarray}}
\newcommand{\beas}{\begin{eqnarray*}}
\newcommand{\eeas}{\end{eqnarray*}}

\usepackage{pgfplotstable}
\usepackage{tikz,pgf,pgfplots}
\usetikzlibrary{calc,shadows}
\usetikzlibrary{pgfplots.groupplots}

\usepackage{ifthen}
\usepackage{mathtools}

\newcommand\Tex{}
\newcommand\PR[2][\Tex]{
\ifthenelse{\equal{#1}{}}{{\mathrm{Pr}}\left(#2\right)}{\ensuremath{{\mathrm{Pr}}_{#1}\left[ #2\right]}}}

\newcommand\EX[2][\Tex]{
\ifthenelse{\equal{#1}{}}{{\mathbb E}\left[#2\right]}{\ensuremath{{\mathbb E}_{#1}\left[ #2\right]}}}

\newcommand\Var[2][\Tex]{
\ifthenelse{\equal{#1}{}}{{\mathrm{Var}}\left[#2\right]}{\ensuremath{{\mathrm{Var}}_{#1}\left[ #2\right]}}}
\newcommand\defeq{\coloneqq}

\renewcommand\M{M} 
\renewcommand\N{N} 
\newcommand\len{L} 
\renewcommand\q{q} 

\newcommand\vf{\mathbf{f}} 

\newcommand\cNM{\lambda} 

\newcommand{\type}{t}

\newcommand\ind[1]{\mathds{1}\left\{#1\right\}}

\newcommand\norm[2][\Tnorm]{\ensuremath{{\left\|#2\right\|}_{#1}}}

\newcommand\Emolseen{1-q_0}

\newcommand\Cbsc{C_{\text{BSC}}} 

\newcommand\Rbsc{R_{\text{BSC}}} 
\newcommand\Rind{R_{\text{index}}} 

\begin{document}

\title{}

\begin{center}

{\bf{\LARGE{
DNA-Based Storage: Models and Fundamental Limits
}}}

\vspace*{.2in}

{\large{
\begin{tabular}{cccc}
Ilan Shomorony$^{\dagger}$ &  Reinhard Heckel$^{\ddagger}$  
\end{tabular}
}}

\vspace*{.2in}

\begin{tabular}{c}
$^\dagger$University of Illinois at Urbana-Champaign \\
$^\ddagger$Technical University of Munich
\end{tabular}

\vspace*{.2in}

\today

\vspace*{.2in}

\begin{abstract}
Due to its longevity and enormous information density, DNA is an attractive medium for archival storage. 
In this work, we study the fundamental limits and trade-offs of DNA-based storage systems by introducing a new channel model, which we call the noisy shuffling-sampling channel. Motivated by current technological constraints on DNA synthesis and sequencing, this model captures three key distinctive aspects of DNA storage systems: (1) the data is written onto many short DNA molecules; (2) the molecules are corrupted by noise during synthesis and sequencing and (3) the data is read by randomly sampling from the DNA pool. We provide capacity results for this channel under specific noise and sampling assumptions and show that, in many scenarios, a simple index-based coding scheme is optimal.
\end{abstract}

\end{center}

\blfootnote{
Parts of this paper were presented at the 2017 and 2019 IEEE International Symposium on Information Theory (ISIT)\cite{heckel_fundamental_2017,shomorony_capacity_2019}.
}

\vspace{-5mm}
\section{Introduction}
\label{sec:intro}
Due to its longevity and enormous information density, and thanks to rapid advances in technologies for writing (synthesis) and reading (sequencing), DNA is on track to become an attractive medium for archival data storage. 
DNA is a long molecule made up of four nucleotides (Adenine, Cytosine, Guanine, and Thymine) and, for storage purposes, can be viewed as a string over a four-letter alphabet. 
While in a living cell a DNA molecule can consist of millions of nucleotides, due to technological constraints, it is difficult and inefficient to synthesize long strands of DNA.
Thus, in practice, data is stored on short DNA molecules which are preserved in a DNA pool and cannot be spatially ordered. 

In recent years, several groups have demonstrated
 working DNA storage systems~\cite{church_next-generation_2012,goldman_towards_2013,grass_robust_2015,yazdi_rewritable_2015,erlich_dna_2016,organick_scaling_2017}.
In these systems, information was 
 stored  on molecules of no longer than one or two hundred nucleotides.
At the time of reading, the information is accessed via state-of-the-art sequencing technologies. This corresponds to (randomly) sampling and reading sequences from the pool of DNA. Sequencing is preceded by several cycles of Polymerase Chain Reaction (PCR) amplification. In each cycle each molecule is replicated by a factor of 1.6-1.8. 
Thus, the proportions of the sequences in the DNA mixture just before sequencing and the probability that a given sequence is read depends on the synthesis method, the PCR steps, and the decay of DNA during storage. 
Finally, sequencing and in particular synthesis of the DNA may lead to insertions, deletions, and substitutions of nucleotides in individual DNA molecules. See~\cite{heckel_channel_2019} for a detailed discussion of the error sources and probabilities for different experimental setups.

Given these constraints, a mathematical model for a DNA storage channel is as follows. Data is written on $\M$ DNA molecules, each of length $\len$. From this multiset of sequences, $\N$ sequences are drawn according to some distribution $Q$, and then perturbed by introducing individual base errors. 
A critical element of this model is that by drawing $\N$ sequences according to some distribution $Q$, the order of the sequences is lost.

\begin{figure}
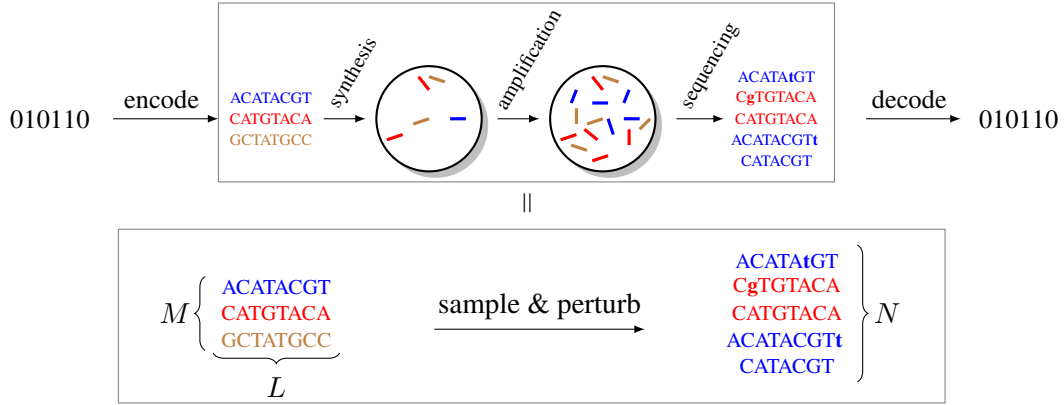

\hspace{1.45cm}
\parbox{\textwidth}{
\figurechannelmodel
}
\caption{
\label{fig:channelmodel}
Channel model for DNA storage systems. 
The input to the channel is a multi-set of $\M$ length-$\len$ DNA molecules and the output is a multi-set of $\N$ draws from the pool of DNA molecules that are perturbed by insertions, substitutions, and deletions (marked as lowercase and boldface letters).
}
\vspace{-2mm}
\end{figure}

The decoder's goal is to reconstruct the information from the multi-set of $\N$ reads. 
Note that the decoder has no information about which molecules were sampled, and in general  
a fraction of the original DNA fragments may never be sampled. 
Our goal is to study the capacity of this channel under different modeling assumptions on the sampling distribution and introduced errors.


\subsection{Contributions}

In this paper we study the fundamental limits of the DNA storage model outlined above.
Our analysis aims to reveal the basic relationships and trade-offs between key design parameters and performance goals such as storage density and reading/writing costs.
Throughout, 
we 
consider the asymptotic regime where $\M \to \infty$.
The main parameter of interest is the storage capacity $C$, defined as the maximum number of bits that can be reliably stored per nucleotide (the total number of nucleotides is $\M\len$).

\paragraph{Capacity in the case of noise-free sequences:}
We start with a channel without errors in the individual sequences. 
Thus, randomness is only introduced through the distribution $Q$, which describes the number of copies we draw from each input sequence. 
According to $Q$, some of the individual sequences might never be drawn and others are drawn many times. 
Our main result for this channel states that if $\lim_{\M \to \infty} \frac{\len}{\log \M} = \beta >1$, then
\al{
C = (1 - q_0) (1- 1/\beta),
\label{eq:cap_noiseless}
}
where $q_0$ is the probability that a given sequences is never sampled. 
Interestingly, our result only depends on the distribution $Q$ through $q_0$, which is the probability that a given sequences is never sampled. 
Moreover, if $\lim_{\M \to \infty} \frac{\len}{\log \M} < 1$, no positive rate is achievable. 
The factor $1 - q_0$ is the loss due to unseen molecules, and $1-1/\beta$ corresponds to the loss due to the unordered fashion of the reading process.

One important implication of our result is that a simple index-based scheme (as commonly used by DNA data storage systems) is optimal; i.e., prefixing each molecule with a unique index incurs no rate loss. 
More specifically, our result shows that indexing each DNA molecule and employing an erasure code across the molecules
is capacity-optimal. 
Furthermore, the capacity in \eqref{eq:cap_noiseless} is only non-trivial if the read length scales as $L = \Theta(\log M)$.
For that reason, throughout the paper we focus on the regime $L = \beta \log M$, where $\beta$ is a positive constant.

Suppose that each sequence is drawn according to a Poisson distribution with mean $\cNM$, so that in expectation $ \cNM M$ sequences are drawn and $\cNM$ can be thought of as the sequencing coverage depth. 
Then, the probability that a sequence is never drawn is $e^{- \cNM}$ and it decays exponentially in the coverage depth. 
For this scenario, our expression for the capacity suggests that practical systems should not operate at a high coverage depth $\N/\M$, as high coverage depth significantly increases the time and cost of reading, but only provides little storage gains. 
Notice that, in order to guarantee that all $M$ sequences are observed at least once, we need $\N = \Omega(\M \log \M)$ \cite{LanderWaterman,motahariDNA}.
When $\M$ is large, it is wasteful to operate in this regime, 
as this only gives a marginally larger storage capacity, but the sequencing costs can be exorbitant. 


\paragraph{Capacity in the case of noisy sequences:}
Our second contribution is an expression for the capacity for the case where the reading of the sequences is noisy. The goal of this second statement is to understand the effect of errors within sequences, in addition to the shuffling and sampling of the sequences. 
We assume that the 
distribution $Q$ with which the sequences are drawn is a simple Bernoulli distribution; i.e., 
a sequence is either drawn once with probability $1-q_0$ or not drawn with probability $q_0$.
Furthermore we focus on substitution errors within sequences.
Thus, we study a noisy shuffling-sampling model where the output sequences are obtained as follows:
(i) each original sequence is drawn with probability $1-q$ and not drawn with probability $q$,
(ii) the drawn sequences are shuffled,
and
(iii) passed through a binary symmetric channel with crossover probability $p$.

In the low-error regime (where $p$ is sufficiently small), the capacity of this noisy shuffling-sampling channel is given by
\al{
C = (1-q) (1 - H(p) - 1/\beta).
\label{eq:bsc_cap}
}
Note that $1-H(p)$ is the capacity of the binary symmetric channel. 
As it turns out, \eqref{eq:bsc_cap} can be achieved by treating each length-$\len$ sequence as the input to a separate BSC and encoding a unique index into each sequence, and using an erasure outer code to protect against the loss of a $q_0$-fraction of the $M$ sequences.
For a large set of parameters $\beta$ and $p$ (described in Section~\ref{sec:bsc}), is capacity-optimal. 
This result provides a theoretical justification for a number of works, starting with~\cite{grass_robust_2015}, which have used a similar coding scheme in real implementations of DNA-based storage systems~\cite{grass_robust_2015,yazdi_rewritable_2015,erlich_dna_2016,organick_scaling_2017,Meiser_Antkowiak_Koch_Chen_Kohll_Stark_Heckel_Grass_2020}.


\subsection{Related literature}
Computer scientists and engineers have dreamed of harnessing DNA's storage capabilities already in the 60s~\cite{neiman_fundamental_1964,baum_building_1995}, and in recent years this idea has developed into an active field of research.
In 2012 and 2013 groups lead by Church~\cite{church_next-generation_2012} and Goldman~\cite{goldman_towards_2013} independently stored about a megabyte of data in DNA. 
In 2015, Grass et al.~\cite{grass_robust_2015} demonstrated that millenia long storage times are possible by protecting the data both physically and information-theoretically, and designed a robust DNA data storage scheme using modern error correcting codes. Later, in the same year, Yazdi et al~\cite{yazdi_rewritable_2015} showed how to selectively access parts of the stored data, and in 2017, Erlich and Zielinski~\cite{erlich_dna_2016} demonstrated that practical DNA storage can achieve very high information densities. In 2018, Organick et al.~\cite{organick_scaling_2017} scaled up these techniques and stored about 200 megabytes of data.

The capacity of a DNA storage system under a related model has been studied in an unpublished manuscript by MacKay, Sayer, and Goldman~\cite{mackay_near-capacity_2015,sayir_challenge_2016}.
In the model in~\cite{mackay_near-capacity_2015}, the input to the channel consists of a (potentially arbitrarily large) set of DNA molecules of fixed length $\len$, which is not allowed to contain duplicates. The output of the channel are $\M$ molecules drawn with replacement from that set. 
The approach in~\cite{mackay_near-capacity_2015} considers coding over repeated independent storage experiments, and computes the single-letter mutual information over one storage experiment. This indicates that the price of not knowing the ordering of the molecules is logarithmic in the number of synthesized molecules, similar to our main result.

The capacity of a DNA storage system under a different model was studied in \cite{erlich_dna_2016}. 
Specifically \cite{erlich_dna_2016} assumes that each DNA segment is indexed which reduces the channel model to an erasure channel. 
While this assumption removes the key aspects that we focus on in this paper, namely that DNA molecules are stored in an unordered way and read via random sampling, \cite{erlich_dna_2016} considers other important constraints, such as homopolymer limitations.

Several recent works have designed coding schemes for DNA storage systems based on this general model, some of which were implemented in proof-of-concept storage systems \cite{church_next-generation_2012,goldman_towards_2013,grass_robust_2015,bornholt_dna-based_2016,erlich_dna_2016}. 
Several papers have studied important additional aspects of the design of a practical DNA storage system.
Some of these aspects include DNA synthesis constraints such as sequence composition  \cite{kiah_codes_2016,yazdi_rewritable_2015,erlich_dna_2016}, the asymmetric nature of the DNA sequencing error channel \cite{gabrys_asymmetric_2015}, 
the need for codes that correct insertion errors \cite{sala_insertions_2016}, 
and the need for techniques to allow random access \cite{yazdi_rewritable_2015}.
The use of fountain codes for DNA storage was considered in both \cite{erlich_dna_2016} and \cite{mackay_near-capacity_2015}. 

Finally, the recent paper~\cite{lenz_upper_2019} studies a related channel and proves a converse on the capacity through a combination of techniques, including ideas from~\cite{heckel_fundamental_2017,shomorony_capacity_2019}.

\section{Problem setting and channel models}
\label{sec:problem}

An $(\M,\len)$ DNA storage code $\C$ is a set of codewords,  
each of which is a list 
$[\x_1,\ldots,\x_M]$
 of $M$ strings of length $\len$,
 together with a decoding procedure. 
The alphabet $\Sigma$ is typically $\{\bA,\bC,\bG,\bT\}$, corresponding to the four nucleotides that compose DNA. 
However, to simplify the exposition we focus on the binary case $\Sigma = \{0,1\}$, and we note that the results can be extended to a general alphabet in a straightforward manner. 
Throughout the paper we use the word molecule or sequence to refer to each of the stored strings of length $\len$ over the alphabet $\Sigma$. 
We study the following general noisy shuffling-sampling channel model:
\begin{enumerate}
\item Given that codeword $[\x_1,\ldots,\x_M] \in \C$ is chosen, each sequence $\x_i$ is sampled a number $\N_i \sim Q$ of times, for some distribution $Q = (q_0,q_1,\ldots)$, where $\q_n = \PR{N_i = n}$ is the probability that $\x_i$ is drawn $n$ many times. 
We let $N = \sum_{i=1}^\M N_i$ be the total number of resulting strings, and we define $\cNM \defeq \EX{\N}/\M = \EX{\N_i}$.
The distribution $Q$ models imperfections in synthesis, sequencing, and a loss of whole sequences during storage (see \cite{heckel_channel_2019} for a detailed discussion on how this distribution looks like for specific choices of sequencing and synthesis technologies). 
\item Each of the resulting $\N$ strings is passed through a discrete memoryless channel.
\item The resulting $\N$ strings are shuffled uniformly at random to yield the output $[\y_1,\ldots,\y_N]$. 
Equivalently, the output of the channel is the (unordered) multi-set of $N$ output sequences $\{\y_1,\ldots,\y_N\}$. 
\end{enumerate}

A decoding function then maps the received sequences $[\y_1,\ldots,\y_N]$ to a message index in $\{1,\ldots,|\C|\}$.
The main parameter of interest of a DNA storage system is the storage density, or the storage rate, 
defined as the number of bits written per DNA base synthesized, i.e., 
\al{
R \defeq \frac{\log|\C|}{\M\len}. \label{eq:Rs}
}
We consider an asymptotic regime where 
$M \to \infty$ and we let $\len \defeq \beta \log M$ for some fixed $\beta$.
As our main results show, $L = \Omega(\log M)$ is the asymptotic regime of interest for this problem.
We say that the rate $R$ is achievable if there exists a sequence of DNA storage codes $\C_\M$ with rate $R$ such that the decoding error probability tends to $0$ as $\M \to \infty$. 


\section{Storage Capacity for the Noise-free Channel}
\label{sec:noisefree}

An important property of DNA storage channels is the fact that the order or the molecules are lost. 
We first focus on this aspect of the channel model by studying the noise-free channel (where all copies are noise-free, i.e., the discrete memoryless channel is just the ``identity channel'').

The main result of this section is the characterization of the storage capacity, given by the following theorem.

\begin{theorem} \label{thm:noisefree}
The storage capacity of the noise-free shuffling-sampling channel is
\al{
C  & = (1-q_0) \left(1 - 1/\beta \right). \label{eq:storagecap}
}
In particular, if $\beta \leq 1$, no positive rate is achievable.
\end{theorem}

The capacity expression in (\ref{eq:storagecap}) can be intuitively understood through the achievability argument.
A storage rate of $R = (1-q_0) \left(1 - 1/\beta \right)$ can be easily achieved by prefixing all the molecules with a distinct tag, which effectively converts the channel to a block-erasure channel. 
More precisely, we use the first $\log \M$ bits of each molecule to encode a distinct index. 
Then we have $\len - \log \M = \len(1 - 1/\beta)$ symbols left per molecule to encode data. 
The decoder can use the indices to remove duplicates and sort the molecules that are sampled. 
This effectively creates an erasure channel, where molecule $i$ is erased if it is not drawn (i.e., $\N_i=0$) which occurs with probability $q_0$. 
Since the expected number of erasures is $\EX{\frac{1}{\M} \sum_{i=1}^\M \ind{\N_i = 0}} = q_0$, we achieve storage rate $\frac{(1-q_0)M(L-\log M)}{ML}=(1-q_0)(1-1/\beta)$.
The surprising aspect of Theorem~\ref{thm:noisefree} is that this simple index-based scheme is optimal.
It is also worth noting that the capacity expression only depends on the sampling distribution $Q$ through the parameter $q_0$, i.e., the fraction of sequences that is \emph{not seen} at the output of the channel.

In order to gain intuition on a practical implication of this theorem, suppose that each sequence is drawn according to a Poisson distribution with mean $\cNM$, so that in expectation $ \cNM M$ sequences in total are drawn and $\cNM$ can be thought of as the sequencing coverage depth. 
Then, the probability that a sequence is never drawn is $e^{- \cNM}$ and the capacity expression becomes
\al{
C = (1-e^{-\lambda})(1-1/\beta).
}
This suggests that practical systems should not operate at a high coverage depth $\N/\M$, as high coverage depth significantly increases the time and cost of reading, but only provides little storage gains, according to our capacity expression. 
Notice that, in order to guarantee that all $M$ sequences are observed at least once, we need $\N = \Omega(\M \log \M)$ \cite{LanderWaterman,motahariDNA}.
When $\M$ is large, it is wasteful to operate in this regime, 
as this only gives a marginally larger storage capacity, but the sequencing costs can be exorbitant. 

The result in Theorem~\ref{thm:noisefree} is flexible to allow different sampling models.
In particular, one can consider separating the PCR amplification performed on each synthesized molecule from the sequencing step.
Since one cannot control the PCR amplification factor precisely, it is reasonable to assume that a molecule $\x$ is first randomly amplified and a total of $A \geq 0$ copies is stored.
If we consider a Poisson sampling model for the sequencing step, the effective coverage depth is $\lambda/E[A]$ (since we are actually sampling from $M E[A]$ molecules).
In this case, the probability that none of the copies of $\x$ is sampled at the output is $\Ex[(e^{-\lambda / \Ex[A]})^A] = \Ex[(e^{(-\lambda / \Ex[A])A}]$.
This can be recognized as the moment-generating function of $A$ evaluated at $-\lambda/ E[A]$.
In particular, when PCR is also modeled as a Poisson random variable with mean $\Ex[A]=\alpha$, 
$\Ex[e^{\theta A}] = e^{\alpha (e^\theta -1)}$, and
the capacity of the resulting noise-free shuffling-sampling channel is
\al{
C = \left(1-e^{-\alpha(1-e^{-\lambda/\alpha})}\right)(1-1/\beta).
}

\subsection{
\label{sec:motivationconverse}
Motivation for Converse}

A simple outer bound can be obtained by considering a genie that provides the decoder with the ``true'' index of 
each sampled molecule.
In other words, $[\x_1,\ldots,\x_\M]$ are the stored molecules, and the decoder observes $[\y_1,\ldots,\y_N]$ and the mapping $\sigma \colon \{1,\ldots, N\} \to \{1,\ldots, M \}$ so that
$\y_j = \x_{\sigma(j)}$.
This converts the channel into an erasure channel with block-erasure probability $\q_0$, which yields
\al{
R \leq 1 - \q_0. \label{eq:simplebound1}
}
It is intuitive that the bound~\eqref{eq:simplebound1} should not be achievable, as the decoder in general cannot sort the molecules and create an effective erasure channel.
However, it is not clear a priori either whether prefixing every molecule with an index is 
optimal.

Notice that one can view the noise-free DNA storage channel as a channel where
the encoder chooses a distribution (or a type) over the alphabet $\Sigma^\len$
and the decoder observes a noisy version of this type where the frequencies are perturbed accoding to $Q$.
From this angle, the question becomes ``how many types $\type \in \Z_+^{2^L}$ with $\|\type \|_1 = M$ can be reliably decoded?'', and restricting ourselves to index-based schemes restricts the set of types to those with $\|\type \|_\infty = 1$; i.e., no duplicate molecules are stored.

While this restriction may seem suboptimal, a counting argument suggests that it is not.
The number of types for a sequence of length $M$ over an alphabet of size $|\Sigma^\len| = 2^L$ is at most $M^{2^L}$ and thus at most
\aln{
\frac{1}{ML} \log M^{2^L} = \frac{2^L \log M}{M \beta \log M}  = \frac{M^\beta}{\beta M} 
}
bits can be encoded per symbol.
We conclude that, if $\beta < 1$, the capacity is $C = 0$.
An actual bound on the rate can be obtained by counting the number of types more carefully.
This is done in the following lemma, which we prove in the appendix.
\begin{lemma} \label{lem:comb}
The number of distinct vectors $\type \in \Z_+^{a}$ with $\|\type \|_1 = b$ is given by
\aln{
\T[a,b] \defeq {a+b-1 \choose b} < \left( \frac{e(a+b-1)}{b} \right)^b.
}
\end{lemma}
Since our types are vectors $\type \in \Z_+^{2^L}$ with $\|\type \|_1 = M$, and $2^L  = 2^{\beta \log M} = M^\beta$, it follows that at most 
\aln{
\frac{1}{ML} \log \left( \frac{e(M^\beta+M-1)}{M} \right)^{M} \leq \frac{M \log (\alpha M^{\beta-1})}{M \beta \log M}
}
bits can be encoded per symbol, for some $\alpha > 1$, and 
\al{
R \leq 1- 1/\beta.  \label{eq:simplebound2}
}
Therefore, if we had a deterministic channel where the decoder  observed \emph{exactly} the $M$ stored molecules, an index-based approach would be optimal from a rate standpoint.
The converse presented in the next section utilizes a more careful genie to show
that the bounds in  (\ref{eq:simplebound1}) and (\ref{eq:simplebound2}) can in fact be combined, implying the optimality of index-based coding approaches.
\subsection{Converse}
\label{sec:converseCounting}

\newcommand\set{\mathrm{set}}

\newcommand\Nsampled{{\tilde N}}

Let $[\x_1,\ldots,\x_M]$ be the $\M$ length-$\len$ molecules written into the channel 
and $[\y_1,\ldots,\y_N]$ be the length-$\len$ molecules observed by the decoder. 
Notice that, whenever the channel output is such that $\y_i = \y_j$ for $i \ne j$, the decoder cannot determine whether both $\y_i$ and $\y_j$ were sampled from the same molecule $\x_\ell$ or from two different molecules that obey $\x_\ell = \x_k, \ell \neq k$. 
In order to derive the converse, we consider a genie-aided channel that removes this ambiguity. 
\begin{figure}[h] 
\vspace{0mm}
	\center
       \includegraphics[width=9cm]{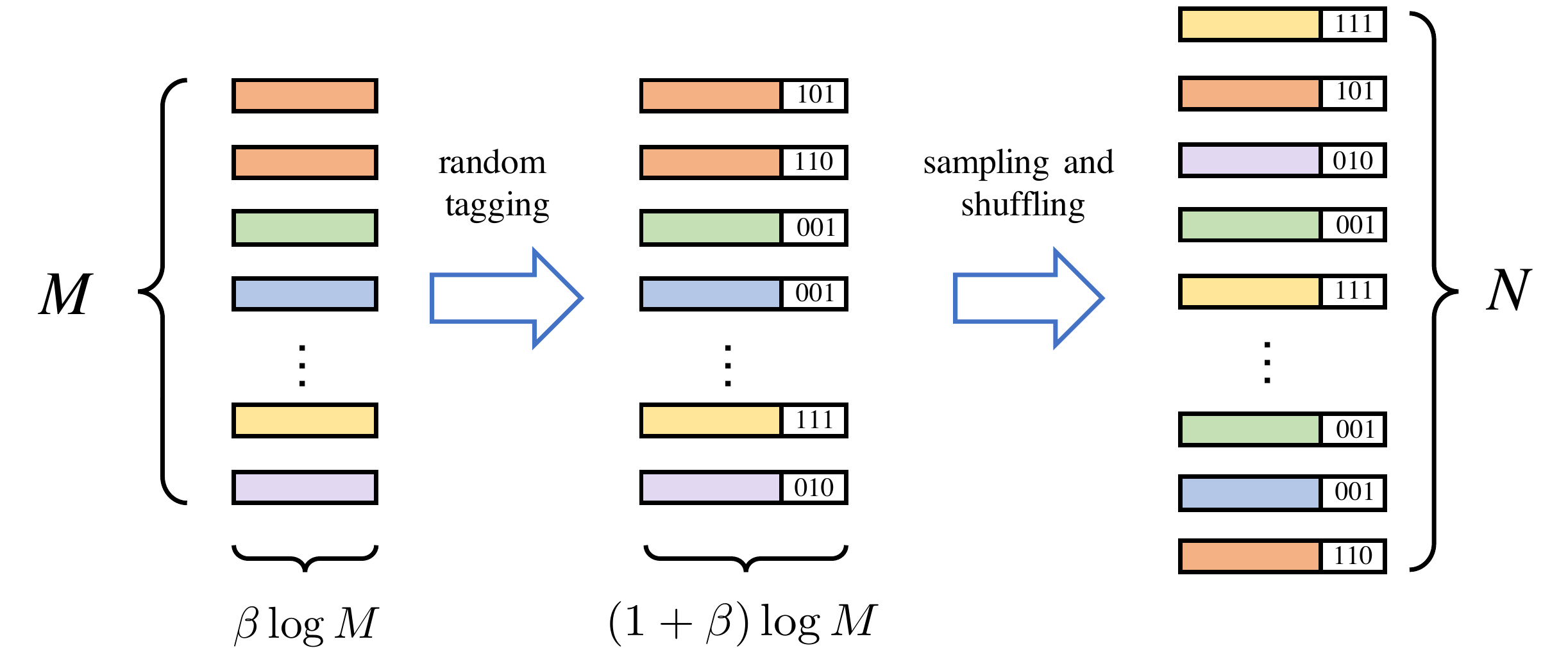} 
       \caption{Genie-aided channel for converse.
       \label{fig:geniechannel}}
\end{figure}
As illustrated in Figure~\ref{fig:geniechannel}, before sampling the $\N$ molecules, the genie-aided channel appends a unique index of length $\log \M$ to each molecule $\x_i$, which results in the set of tagged molecules $\{(\x_i,z_i)\}_{i=1}^\M$. 
We emphasize that the indices $z_i$ are all unique, and are chosen randomly and independently of the input sequences $\{\x_i\}_{i=1}^\M$. 
Notice that, in contrast to the naive genie discussed in Section~\ref{sec:motivationconverse}, this genie does \emph{not} reveal 
the index $i$ of the molecule $\x_i$ from which $\y_\ell$ was sampled.
Therefore, the channel is not reduced to an erasure channel, 
and intuitively the indices are only useful for the decoder to determine whether two equal samples $\y_\ell = \y_k$ came from the same molecule or from distinct molecules. 

The output of the genie-aided channel, denoted by $\{(\y_i,z_{\sigma(i)})\}_{i = 1}^\N$, is then obtained 
by sampling from the set of tagged molecules $\{(\x_i,z_i)\}_{i=1}^\M$, in the same way as the original channel samples the original molecules.
The mapping $\sigma : [1:N] \to [1:M]$ is such that $\y_i$ was sampled from $\x_{\sigma(i)}$.
Notice that the actual mapping $\sigma$ is not revealed to the decoder.

It is clear that any storage rate  achievable in the original channel can be achieved on the genie-aided channel, as the decoder can simply discard the indices, or stated differently, the output of the original channel can be obtained from the output of the genie-aided channel. 

Notice that $\{(\y_i,z_{\sigma(i)})\}_{i = 1}^\N$ is in general a multi-set.
We let $\set( \{(\y_i,z_{\sigma(i)})\}_{i = 1}^\N )$ be the set obtained from $\{(\y_i,z_{\sigma(i)})\}_{i = 1}^\N$ by removing any duplicates. 
Then $\set( \{(y_i,z_{\sigma(i)})\}_{i = 1}^\N )$ is a sufficient statistic for $\{\x_i\}_{i=1}^\M$ since all tagged molecules are distinct objects, and sampling the same tagged molecule $(\x_i,z_i)$ does not yield additional information on $\{\x_i\}_{i=1}^\M$. 
More formally, conditioned on $\set( \{(\y_i,z_{\sigma(i)})\}_{i = 1}^\N )$, $\{\x_i\}_{i=1}^\M$ is independent of the genie's channel output $\{(\y_i,z_{\sigma(i)})\}_{i = 1}^\N$. 

Next, we define the frequency vector $\vf \in \Z_+^{\M^{\beta}}$ 
(note that $|\Sigma^{\len}| = 2^{\beta \log \M} = \M^\beta$) 
obtained from $\set( \{(y_i,z_{\tilde i})\}_{i = 1}^\N )$ in the following way. 
The entry of $\vf$ corresponding to $\y$, for $\y \in \Sigma^{\len}$, is given by 
\aln{
\vf[\y] 
\defeq \left| \left\{ (\y_j,z_{\sigma(j)}) \in \set( \{(\y_i,z_{\sigma(i)})\}_{i = 1}^\N )
\colon  \y_j = \y \right\} \right|. 
}
Hence, $\vf$ is essentially a histogram that counts the number of occurrences of $\y \in \Sigma^L$ in the set of tagged molecules $\{(\y_i,z_{\sigma(i)})\}_{i = 1}^\N$.
Notice that the entries of $\vf$ can take values greater than one.

Since $\set( \{(\y_i,z_{\sigma(i)})\}_{i = 1}^\N )$ is a sufficient statistic for $\{\x_i\}_{i=1}^\M$
and the tags added by the genie were chosen at random and independently of $\{\x_i\}_{i=1}^\M$, it follows that $\vf$ is also a sufficient statistic for $\{\x_i\}_{i=1}^\M$. 
Hence, we can view the (random) frequency vector $\vf$ as the output of the channel without any loss. 
Notice that $|\set( \{(\y_i,z_{\sigma(i)})\}_{i = 1}^\N )| =  \| \vf \|_1$, and in expectation we have 
$\EX{ \| \vf \|_1 / \M } = \frac{1}{\M} (1 - \q_0)$. 
Furthermore, the following lemma asserts that $\norm[1]{\vf}$ does not exceed its expectation by much.

\begin{lemma} \label{lem:conc}
For any $\delta > 0$, the frequency vector $\vf$ at the output of the genie-aided channel satisfies
\aln{
\PR{
 \frac{\| \vf \|_1}{\M} >
1 - \q_0
 + \delta 
 } 
 \to 0, \text{ as } \M \to \infty. 
}
\end{lemma}
\begin{proof}
Note that the number of distinct fragments that have been drawn is 
\[
\frac{\norm[1]{\vf}}{\M} = \frac{1}{\M} \sum_{i=1}^\M \ind{ \N_i > 0 }.
\]
Since $\ind{ \N_i > 0 }$ are independent random variables with expectation $1-\q_0$, 
Hoeffding's inequality yields
\[
\PR{ \frac{\norm[1]{\vf}}{\M} \geq (1-\q_0)
+ \delta} 
\leq e^{-2\M \delta^2},
\]
which concludes the proof\footnote{
An analogue of Lemma~\ref{lem:conc} can be proved for a different sampling model, which we describe in Appendix~\ref{app:lemma2}.
}.
\end{proof}


We now append the coordinate
$
f_0 = ( 1-q_0 + \delta) \M - \| \vf \|_1
$
to the beginning of $\vf$ to construct $\vf' = (f_0,\vf)$.
Notice that when $\| \vf \|_1 \leq (1-q_0 + \delta)\M$ (which by Lemma~\ref{lem:conc} happens with high probability), we have $\| \vf' \|_1 = ( 1-q_0 + \delta)\M$.

Fix $\delta > 0 $, and define the event $\E =  \{ \| \vf \|_1 > (1-q_0 + \delta)\M \}$ with indicator function $\one_\E$.
By Lemma~\ref{lem:conc}, $\PR{\E} \to 0$ as $\M \to \infty$. 
Consider a sequence of codes $\{\C_\M\}$ with rate $R$ and vanishing error probability.
If we let $W$ be the message to be encoded, chosen uniformly at random from $\{1,\ldots,2^{\M \len R}\}$,
from Fano's inequality we have
\al{
\M \len R_s  
&= 
H(W)
=
I(W; \vf') + H(W| \vf') \nonumber \\
&\leq H(\vf') + 1 + P_e \M \len R_s,
}
where $P_e$ is the probability of a decoding error, which by assumption goes to zero as $M \to \infty$. 
We can then
upper bound the achievable storage rate $R$ as
\al{
\M & \len R_s (1-P_e) 
\leq H (\vf' )  +1 \leq  H\left( \vf', \one_\E \right) +1 \nonumber \\
&\leq \PR{ \E } H\left( \left. \vf'\, \right|  \E  \right) + \PR{ \bar \E } H \left( \left. \vf' \, \right|  \bar \E  \right) + H(\one_\E) + 1 \nonumber \\
&\leq \PR{ \E } \log \T [\M^\beta+1,\M] \nonumber \\ 
&\quad \quad + \log \T [\M^\beta+1, (\Emolseen + \delta)\M] + 2, \label{eq:rsbound}
}
where $\T[a,b]$ is the number of vectors $x \in \Z_+^a$ with $\| x \|_1 = b$.
An application of Lemma~\ref{lem:comb} yields 
\aln{
\log \T & [\M^\beta+1, (\Emolseen + \delta)\M] \\
& \leq (\Emolseen + \delta)\M \log \left( e + \frac{e M^{\beta-1}}{(1-e^{-\cNM} + \delta)} \right) \\ 
& \leq (\Emolseen + \delta)\M \log \left( \alpha M^{\beta-1} \right) \\
& \leq (\Emolseen + \delta)\M  [ (\beta-1) \log \M + \log \alpha ],
}
where $\alpha$ is a positive constant. 
Analogously, we obtain 
\aln{
\log \T [\M^\beta+1,\M]  \leq \M ( (\beta -1) \log \M + \log \alpha ). 
}
Dividing (\ref{eq:rsbound}) by $\M\len$ and applying the bounds above yields
\aln{
R (1 - P_e) &\leq \Pr( \E ) \frac{ \M [(\beta -1) \log \M + \log \alpha ]}{\M \len} \\
& + \frac{( \Emolseen + \delta)\M  [ (\beta-1) \log \M + \log \alpha ]}{\M \len} + \frac{2}{\M\len} \\
& \leq \Pr( \E ) \left( \frac{\beta - 1}{\beta} +\frac{\log \alpha}{\beta \log M}  \right) \\ 
&  + (\Emolseen + \delta) \left(1-\frac{1}{\beta}  + \frac{\log \alpha}{\beta \log \M}\right)  + \frac{2}{\M\len}.
}
Finally, letting $\M \to \infty$ yields
\aln{
R \leq ( \Emolseen + \delta ) \left(1- 1/\beta \right),
}
since $\Pr(\E) \to 0$ by Lemma~\ref{lem:conc}.
Since $\delta > 0$ can be chosen arbitrarily small, this concludes the converse proof of Theorem~\ref{thm:noisefree}.




\section{The noisy shuffling-sampling channel}
\label{sec:bsc}





Next, we study the effect of errors within sequences, in addition to the shuffling and sampling of the sequences.
Instead of the general sampling distribution $Q$ considered in Section~\ref{sec:noisefree}, we now focus on a simple choice of sampling distribution and let $Q$ be distributed as $\text{Bernoulli}(1-q)$. 
Hence, a sequences is either drawn never or once, with the corresponding probabilities given by $\Pr(N_i= 0) = q$ and $\Pr(N_i=1)=1-q$, for $i = 1,\ldots,M$.
Moreover, we assume that the molecules are all corrupted by a BSC with error probability $p$.
We refer to this channel as the noisy shuffling-sampling channel. 

\subsection{The capacity of the noisy shuffling-sampling channel}


As in the error-free shuffling-sampling channel considered in Section~\ref{sec:noisefree}, we again consider a simple index-based coding scheme. 
As we will show, for a large set of parameters $p$ and $\beta$, this scheme turns out to be capacity-optimal.

We consider a erasure-correcting code with block length $\M$ and rate $(1-q)$, where each symbol is itself a binary string of length
$\len(1-H(p) - 1/\beta-\epsilon)$, for some small epsilon.
This code will be used as an outer code.
Our inner code will be a code designed for a BSC with codewords of length $\len$ and rate $R_{\text{BSC}} = 1-H(p) - \ep$.
We first encode the information using the outer code, which yields $\M$  symbols, which are binary strings of length 
\aln{
\len(1-H(p) - 1/\beta-\epsilon)
= L R_{\text{BSC}} - \log M.
}
We take each symbol, add a unique binary index of length $\log \M$ and encode the resulting sequence using the BSC code, which yields $M$ length-$L$ sequences.

With this scheme, we encode a total of $(1-q)M(\len \Rbsc - \log M)$ data bits, with a data rate of 
\al{
\frac{(1-q)\M \left( \len \Rbsc - \log \M\right)}{\M \len} = (1-q)(\Rbsc - 1/\beta). \label{eq:schemerate}
}
Since $\ep > 0$ can be chosen arbitrarily small,
this scheme achieves a rate arbitrarily close to 
\al{
\Rind = (1-q)(1 - H(p) - 1/\beta). \label{eq:lower}
}

Strictly speaking, the simple index-based scheme described above needs to be slightly modified to account for the fact that, if an inner codeword is decoded in error (which occurs with a small probability) its unique index will also be decoded in error, likely causing an ``index collision'' with another correctly decoded inner codeword.
Such a collision effectively creates two erasures.
Moreover, there exists an even smaller probability that two inner codewords are decoded in error in a way that the true indices are swapped. 
Such an event may not be detected at the decoder side based on the set of decoded indices.
Notice, however, that since the error probability of the inner code goes to zero as $M \to \infty$, these events are much rarer than the erasures caused by the sampling distribution $Q$.
It is straightforward to show that by considering an outer code with rate $1-q - \epsilon_2$, for an arbitrarily small $\epsilon_2$, these additional small-probability events can be accounted for.
Hence, (\ref{eq:lower}) is a lower bound to the capacity $C$ of the noisy shuffling-sampling channel.


On the other hand, the result from Section~\ref{sec:noisefree}, with $Q \sim {\rm Ber}(1-q)$ implies that $C \leq (1-q)(1-1/\beta)$, since the error-free shuffling-sampling channel cannot be worse than the noisy shuffling-sampling channel.
Furthermore, a simple genie-aided argument where the decoder observes the shuffling map can be used to establish that $C \leq (1-q)\Cbsc$, where $\Cbsc = 1-H(p)$ is the capacity of a BSC with crossover probability $p$.
Hence, a capacity upper bound is given by
\al{
C \leq (1-q) \min\left[ 1- H(p), 1-1/\beta \right]. \label{eq:upper}
}
Our main result improves on the upper bound in (\ref{eq:upper}), and establishes that for parameters $(p,\beta)$ in a certain regime, the lower bound in equation~\eqref{eq:lower} is the capacity.

%
%

\begin{theorem}
\label{thm:bsc}
The capacity of the noisy shuffling-sampling channel is
%
%
\al{
C = (1-q)(1 - H(p) - 1/\beta), \label{eq:capacity}
}
as long as $p<1/4$ and $1-H(2p)- 2/\beta > 0$.
Moreover, if $\beta \leq 1$, the capacity is $C = 0$.
\end{theorem}

The set of parameters $(p,\beta)$ such that $1-H(2p)- 2/\beta > 0$ and $p<1/4$ is the blue region in Figure~\ref{fig:capacity}.
In particular, (\ref{eq:capacity}) holds if $p \leq 0.1$ and $\beta \geq 6.4$, or if $p < 0.01$ and $\beta \geq 2.35$.

%
%
%
%

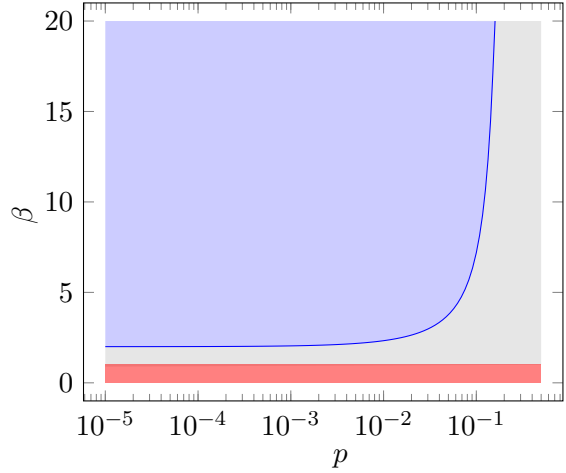
\begin{figure}[h]
\centering 
\begin{tikzpicture}[scale=0.93]

    \begin{semilogxaxis}[enlargelimits=0.05,xlabel=$p$,ylabel=$\beta$]
\addplot[name path=f,color=blue,mark=none,draw=blue!] table[x index=0,y index=1]{./curve.dat};   
   
   \path[name path=axis] (axis cs:0.00001,20) -- (axis cs:0.1,20);
   \addplot[blue!20] fill between[of=f and axis];

   \path[name path=axis2] (axis cs:0.00001,0.9) -- (axis cs:0.5,1) -- (axis cs:0.5,20) -- (axis cs:0.1,20);
   \addplot[gray!20] fill between[of=f and axis2];    
   
   \addplot[name path=f, fill=red, fill opacity=0.5, draw=none, mark=none]
coordinates {
    (0.00001, 0)
    (0.00001, 1)
    (0.5, 1)
    (0.5, 0)
};

\end{semilogxaxis}
\end{tikzpicture}  
\vspace{2mm}
\caption{\label{fig:capacity}
Parameter regions for which the capacity is characterized.
The capacity in the blue region is given by $C = (1-q)(1 - H(p) - 1/\beta)$, and the capacity in the red region (i.e., for $\beta < 1$) is $0$.
In the gray region, it is still unknown.}
\vspace{-2mm}
\end{figure} 

\subsection{Converse}
\label{sec:converse}


To derive the converse, we view the input to the channel as a binary string of length $\M \len$, denoted by
\aln{
X^{\ML} = \left[ X_1^\len, X_2^\len, \ldots , X_\M^\len \right] 
\in \{0,1\}^{\M\len}
}
or, equivalently, $\M$ strings of length $\len$ concatenated to form a single string of length $\M \len$.
Similarly, the output of the channel is
\aln{
Y^{\NL} = \left[ Y_1^\len, Y_2^\len, \ldots, Y_\N^\len \right] 
\in \{0,1\}^{\N\len},
}
where $N = \sum_i N_i$.
It is useful to define a vector $S^N \in \{1,\ldots,\M\}^\N$ indicating the input string from which each output string was sampled. 
Furthermore, we let $Z^{\NL}= \left[ Z_1^\len, \ldots, Z_\N^\len \right]$ be the random binary error pattern created by the BSC on the $N$ non-deleted strings.
We can now define the input-output relationship
\al{
Y_k^\len &= X^{\len}_{S(k)} \oplus Z^{\len}_{k}, 
\quad
\text{for $k=1,\ldots,N$},
\label{eq:inputoutput}
}
where $\oplus$ indicates elementwise modulo $2$ addition. 
Note that the $N_i$'s are fully determined by the vector $S^\N$ since $N_i = | \{ i \colon S(k) = i\} |$.
Also note that, since $Q \sim {\rm Ber}(1-q)$, $N \leq M$ with probability $1$.

Consider a sequence of codes for the noisy shuffling-sampling channel with rate $R$ and vanishing error probability.
Let 
$
X^{\ML} = \left[ X_1^\len, X_2^\len, \ldots, X_\M^\len \right]
$
be the input to the channel when we choose one of the $2^{\ML R}$ codewords from one such code uniformly at random, and 
$
Y^{\NL} = \left[ Y_1^\len, Y_2^\len, \ldots, Y_\M^\len \right]
$
be the corresponding output.
First note that
\begin{align*}
MLR 
&= H\left(X^{\ML}\right) =I\left(X^{\ML};Y^{\NL}\right) + \ML \epsilon_{\M},
\end{align*}
where $\epsilon_{\M} \to 0$ as $\M \to \infty$ by Fano's inequality.
Then,
\al{
ML(R-\epsilon_{\M}) & = H\left(Y^{\NL}\right) -  H\left(Y^{\NL}| X^{\ML} \right) \nonumber \\
& = H\left(Y^{\NL}\right) -  H\left(S^N,Z^{\NL},Y^{\NL} | X^{\ML}\right)
+ H\left(S^N,Z^{\NL}|X^{\ML},Y^{\NL}\right) \nonumber \\ 
& = H\left(Y^{\NL}\right) -  H\left(S^N,Z^{\NL},Y^{\NL} | X^{\ML}\right)
 + H\left(S^N|X^{\ML},Y^{\NL}\right) \label{eq:bound1}
}
The last equality follows by noticing that, given $(S^N,X^{\ML},Y^{\NL})$, 
one can compute
$
Z^{\len}_{k} = Y_k^\len \oplus X^{\len}_{S(k)} 
$
for $1 \leq k \leq N$, 
and thus  $H\left(Z^{\NL}|X^{\ML},Y^{\NL},S^N\right) = 0$.
Since $N$ is a function of $S^N$, and $S^N$ and $Z^{NL}$ are independent of $X^{ML}$, the second term in (\ref{eq:bound1}) can be expanded as
\al{
H\left(S^N,Z^{\NL},Y^{\NL} | X^{\ML}\right)
&= H\left(S^N, N \right) + H\left(Z^{\NL}| S^N, N \right) + H\left(Y^{\NL} | X^{\ML}, S^N,Z^{\NL}\right) \nonumber\\
&\eqnum H(N) + H\left(S^N |N  \right) + H\left(Z^{\NL}| N \right) 
\nonumber\\
\quad &\eqnum H(N) + \sum_{n=1}^M \Pr(N = n) \left[ \log \frac{M!}{(M-n)!} + nL H(p) \right] \nonumber \\
\quad &\eqnum  \sum_{n=1}^M \Pr(N = n)  \left(n \log M  + n L H(p) \right) + o(ML) \nonumber \\
\quad & = \Ex[N] \M \left(  \log \M + \len H(p) \right) + o(\ML) \nonumber \\
\quad & = (1-q)\left[ \M \log \M + \ML H(p) \right] + o(\ML). \label{eq:bound2}
}
For $(i)$ we used that $H\left(Y^{\NL} | X^{\ML}, S^N,Z^{\NL}\right) = 0$ since $Y^{\NL}$ is determined by $X^{\ML}, S^N,Z^{\NL}$, 
and $(ii)$ follows from the fact that, given $N=n$, $S^N$ is chosen uniformly at random from all vectors in $\{1,\ldots,M\}^n$ with distinct elements.
For $(iii)$, we used the fact that, from Stirling's approximation,
\aln{
\log \frac{M!}{(M-n)!} & = M\log M - (M-n) \log (M-n) +  o(ML) \\
& = M\log M - (M-n) \log M + (M-n) \log \frac{M}{M-n} +  o(ML) \\
& = n \log M + (M-n) \log \frac{M}{M-n} +  o(ML),
}
and, by Jensen's inequality,
\aln{
0 & \leq \sum_{n>0} \Pr(N = n) (M-n) \log \frac{M}{M-n}  \\ 
& \leq (M- \Ex[N]) \log \frac{M}{(M- \Ex[N])} = (1-q) M \log 1/q = o(ML).
}
%
%
%
%
In order to finish the converse, we need to jointly bound the first and third terms in equation~(\ref{eq:bound1}).
This step is summarized in the following lemma:
\begin{lemma} \label{lem1}
If $\beta$ and $p < 1/4$ satisfy
\al{
1-H(2p)- 2/\beta > 0, \label{eq:condition}
}
then it holds that 
\aln{
H\left(Y^{\NL}\right) + H\left(S^N|X^{\ML},Y^{\NL} \right) \leq (1-q)\ML + o(\ML).
}
\end{lemma}
The parameter regime $(p,\beta)$ for which~\eqref{eq:condition} holds is the regime in which our capacity expression holds, illustrated in Figure~\ref{fig:capacity}. 
Combining (\ref{eq:bound1}), (\ref{eq:bound2}) and Lemma \ref{lem1}, we have
\aln{
ML(R-\epsilon_{\M}) 
&  \leq (1-q)\left( \ML  - \ML H(p) - M \log \M \right) +  o(\ML).
}
Dividing by $\ML$ and letting $\M \to \infty$ yields the converse.

\subsection{Intuition for Lemma~\ref{lem1}}

In order to discuss the intuition for Lemma~\ref{lem1} let us focus on the case $q=0$; i.e., none of the molecules are lost at the output.
In this case, $N=M$, and $S^M$ is chosen uniformly at random from all permutations of $[1,...,M]$.
If we naively bound each entropy term separately, we obtain 
\aln{
H\left(Y^{ML}\right) + H\left(S^N|X^{\ML},Y^{ML}\right) \leq \ML + \M \log \M. 
}
However, intuitively, the bound $H\left(S^M|X^{\ML},Y^{ML}\right) \leq \M \log \M$ is too loose because, as we argue below, 
if the entropy term $H\left(Y^{ML}\right)$ is large then we expect $H\left(S^M|X^{\ML},Y^{\NL}\right)$ to be small and vice versa.

To see this, first note that from $X^{\ML} = x^{\ML}$ and $Y^{ML} = y^{ML}$, one can estimate the permutation $S$ that maps each output string to the corresponding input string, $S^M$, by finding, for each $y_i^\len$, the $x_j^\len$ that is closest to it and setting $S(i) = j$.
This is a good estimate if no other $x_k^\len$ is close to $x_j^\len$. 
There are two regimes, illustrated in Figure~\ref{fig:points}, one where $S^N$ can be estimated well and one where it cannot.
\begin{figure}[ht] 
	\center
       \includegraphics[width=0.45\linewidth]{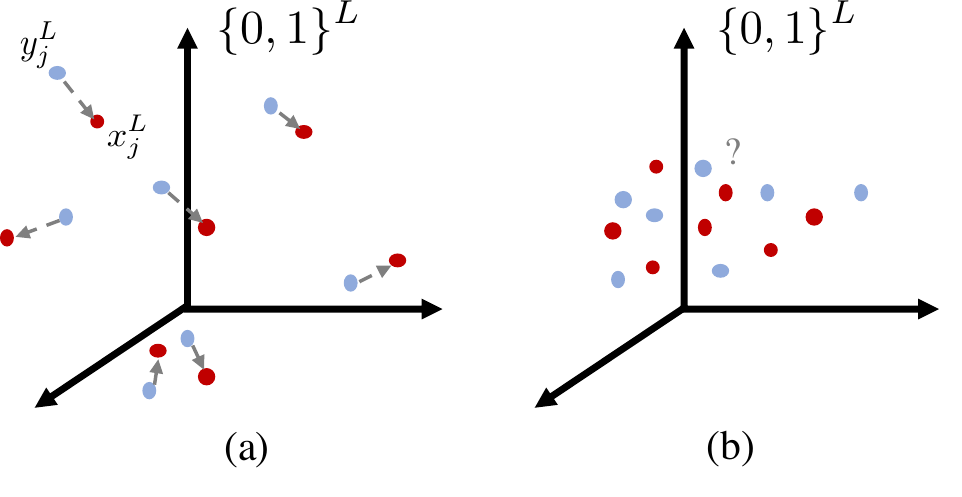} 
              \vspace{-3mm}
       \caption{Two opposite scenarios for estimating $S^N$ from $\left(X^{\ML},Y^{\NL}\right)$.\label{fig:points}}
\vspace{-2mm}
\end{figure}
In the first regime, the strings $x_1^\len,\ldots,x_\M^\len$ are all sufficiently distant from each other (in the Hamming sense).
Hence, the maximum likelihood estimate of $S^N$ given $X^{\ML} = x^{\ML}$ and $Y^{\NL} = y^{\ML}$ is ``close'' to the truth and we expect 
$H\left(S^N|X^{\ML} = x^{\ML},Y^{\NL} = y^{\ML}\right)$ to be small.
In the second regime, illustrated in Fig.~\ref{fig:points}(b), many of the sequences $x_1^\len,\ldots,x_\M^\len$ are close to each other.
So we have less information about $S^N$, and $H\left(S^N|X^{\ML} = x^{\ML},Y^{\NL} = y^{\ML}\right)$ may be large.

On the other hand, the term $H\left(Y^{\NL}\right)$ is maximized if the sequences $\left\{X_i^\len\right\}$ are independent and if their values are uniformly distributed in $\{0,1\}^\len$.
Hence, in order for $H\left(Y^{\NL}\right)$ to be large, we expect to be in the regime in Fig.~\ref{fig:points}(a) instead of the regime of Fig.~\ref{fig:points}(b).
This leads to a tradeoff of the terms $H\left(Y^{\NL}\right)$ and $H\left(S^N|X^{\ML},Y^{\NL}\right)$, which we exploit to prove Lemma~\ref{lem1}.
The detailed proof, which considers the general case where $q\ne 0$, is presented in the appendix.

\section{Discussion}

In this paper we studied the fundamental limits of models of DNA-based storage systems, characterized by random sampling of the input sequences, shuffling, and perturbing them. 
Specifically, we considered a large class of channel models that capture a range of specific instances of DNA storage channels, specified by choices of synthesis, sequencing, and DNA handling technologies.
We focused our analysis on two cases: (1) the error-free shuffling-sampling channel for an arbitrary sampling distribution $Q$ and (2) the noisy shuffling-sampling channel where $Q \sim {\rm Ber}(1-q)$ and the noisy channel is a BSC.
In both cases we proved that a simple index-based scheme is capacity optimal, with the caveat that, for the noisy shuffling-sampling channel, the capacity expression in \eqref{eq:capacity} only holds for the parameter regime of $(p,\beta)$ in the blue region of Figure~\ref{fig:capacity}, and most importantly only holds in the low-error regime.

While the parameter regime in Figure~\ref{fig:capacity} is arguably the most relevant one, an interesting question for future work 
 is whether expression~(\ref{eq:capacity}) is still the capacity of the BSC-shuffling channel if $\beta$ and $p$ do not satisfy (\ref{eq:condition}) (i.e., the gray region in Figure~\ref{fig:capacity}).
Notice that this is a high-noise, short-block regime, and it is reasonable to postulate that coding across the different sequences can be helpful and an index-based approach might not be optimal. 
Another natural question raised by Theorem~\ref{thm:bsc} is whether a similar capacity expression holds for different noisy channels, including corruptions induced by deletions and insertions.


\subsection{General symmetric channels}

Recalling that the capacity expression for the noisy shuffling-sampling channel given by \eqref{eq:capacity} is $(1-q)(\Cbsc -1/\beta)$, it is natural to ask whether for a different sequence-level noisy channel with capacity $C_{\text{noisy}}$, the corresponding noisy shuffling-sampling channel has capacity $(1-q)(C_{\text{noisy}}-1/\beta)$.
As it turns out, the converse proof in Section~\ref{sec:converse} can be extended to the class of \emph{symmetric} discrete memoryless channels (those channels are described in~\cite[Chapter  7.2]{cover2012elements}). 

Specifically, consider a noisy shuffling-sampling channel with sampling $Q \sim {\rm Ber}(1-q)$, and a symmetric discrete memoryless channel (SDMC) with output alphabet $\mathcal{Y}$.
It is then straightforward to generalize the converse proof in Section~\ref{sec:converse} to establish the following result.

\begin{theorem}
\label{thm:extension}
If $\beta$ is large enough, 
the capacity of the SDMC shuffling-sampling channel is given by
\al{
C = (1-q)(C_{\text{SDMC}} - 1/\beta). \label{eq:capacity2}
}
Moreover, if $\beta \leq \log|\mathcal{Y}|$, $C = 0$.
\end{theorem}

For symmetric channels, capacity is achieved by making the distribution of the output uniform, which allows an analogous result to Lemma~\ref{lem1} to be obtained. 
How large $\beta$ needs to be for this statement to hold, depends on the specific channel transition matrix.

Beyond symmetric channels, new converse techniques must be developed in order to characterize the capacity of the corresponding noisy shuffling-sampling channels.

\subsection{Storage-Recovery Tradeoff} 
\label{sec:tradeoff}

Most studies on DNA-based storage emphasize the storage rate (or storage density), while sequencing  costs are disregarded.
From a practical point of view, it is important to understand, for a given storage rate, how much sequencing is required for reliable decoding, as this determines the time and cost required for retrieving the data.
Thus, characterizing the 
storage-recovery trade-off is of practical relevance relevance.

One way to do this is to consider, in addition to the storage rate, the \emph{recovery rate}, defined as the number of bits recovered per DNA base sequenced,
\al{
R_r \defeq \frac{\log|\C|}{\N\len}. \label{eq:Rr}
}
In a practical setting, one can control the amount of sequencing performed, typically specified in terms of the coverage depth $N/M$.
If we consider the error-free shuffling-sampling channel from Section~\ref{sec:noisefree}, in the case
where $Q$ is a Poisson distribution with mean $\cNM$, then $\cNM = N/M$ is the coverage depth, and one would like to choose a value of $\lambda$ that achieves a good trade-off between storage rate and recovery rate.

If we let $R_s$ be the storage rate (previously just $R$, see \eqref{eq:Rs}), from Theorem~\ref{thm:noisefree} and 
the fact that $R_s = \cNM R_r$, the $(R_s,R_r)$ feasibility region can be fully characterized.

\begin{cor}
For the error-free shuffling-sampling channel with $Q \sim {\rm Pois}(\cNM)$, rates $(R_s,R_r)$ are achievable if and only if, for some $c > 0$,
\al{
R_s  & \leq (1-e^{-\cNM})\left(1 - 1/\beta \right),\nonumber \\
R_r & \leq \frac{1-e^{-\cNM}}{\cNM}\left(1 - 1/\beta \right). 
\nonumber
}
\end{cor}


This region is illustrated in Figure~\ref{fig:region}.
This tradeoff suggests that a good operating point would be achieved by not trying to maximize the storage rate (which technically requires $\cNM \to \infty$).
Instead, by using some modest coverage depth $\cNM=1,2,3$, most of the storage rate ($63 \%,86\%,95\%$, respectively) can be achieved.
This is somewhat in contrast to what has been done in the practical DNA storage systems that have been developed thus far, where the decoding phase utilizes very deep sequencing.

To be concrete, suppose that we are interested in minimizing the cost for storing data on DNA. 
Synthesis costs are currently larger than sequencing costs by about a factor $q = 10,000$-$100,000$. 
Thus, if our goal is to minimize the cost for synthesizing and sequencing a given number of bits in DNA, the cost is proportional to $q/R_s + 1/R_r = \frac{q+\cNM}{1-e^{-\cNM}}$. 
This quantity can be maximized over $\cNM$, to obtain the optimal cost per bit stored.
For example, for $q=10000$, $\cNM \approx 9.2$. 
Moreover, one might be interested in optimizing other quantities such as reading time or considering a scenario where the data is read more than once.

\begin{figure} [h]
	\center
       \includegraphics[width=5.6cm]{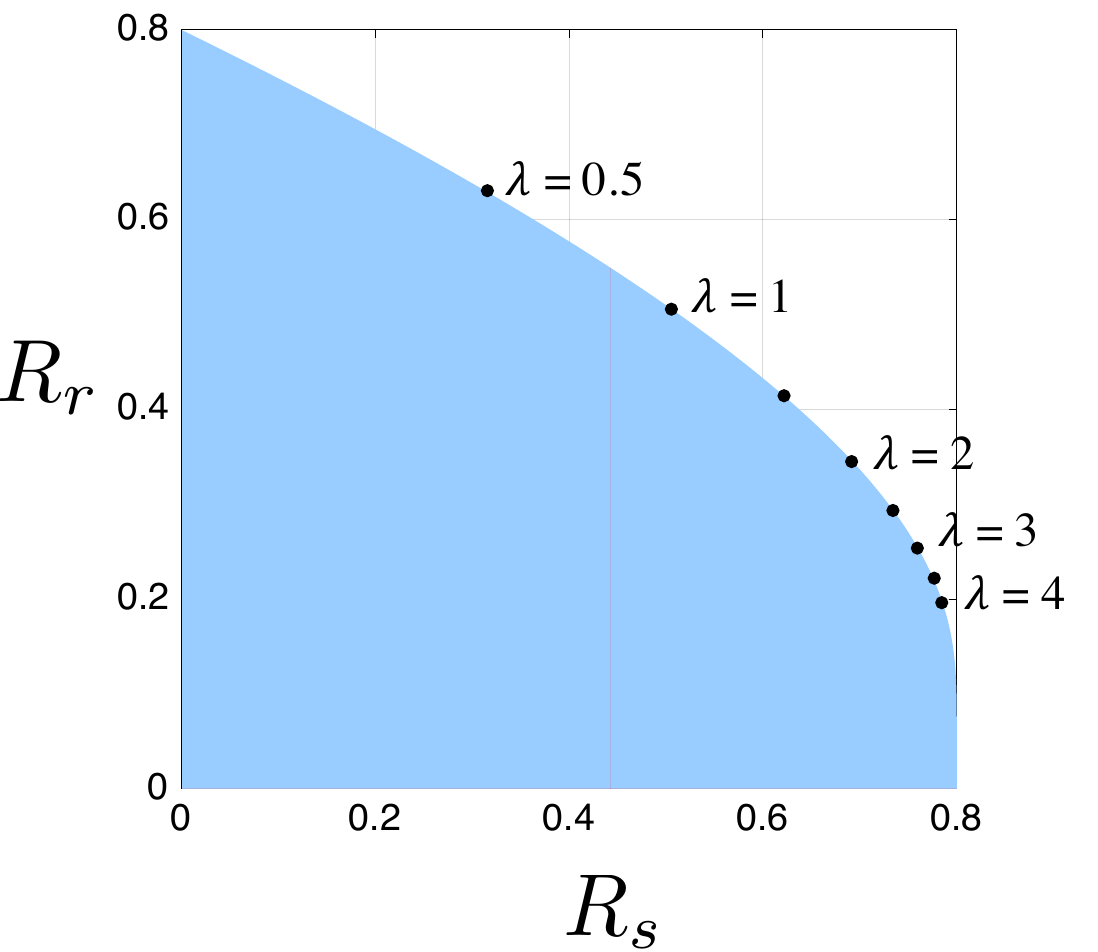} 
       \caption{$(R_s,R_r)$ feasibility region for $\beta = 5$. \label{fig:region}}
\end{figure}


\subsection{Storing data on short molecules}

Throughout this paper, we focused on the regime $\len = \beta \log \M$, with $\beta \geq 1$. 
For $\beta \leq 1$, no positive rate can be achieved (as shown by Theorem~\ref{thm:noisefree}).
However, motivated by the fact that it is in general much easier to synthesize very short sequences of DNA than longer ones, it is interesting to ask whether with very short sequences, it is still possible to build useful DNA storage systems.

Towards this goal, in this section we briefly discuss how fast the rate tends to zero in the regime when $\beta\leq1$. 
Notice that, when $\beta \leq 1$, the total number of distinct molecules of length $L = \beta \log M$ is $2^{\beta \log M} = M^\beta < M$.
Hence, it is impossible to write $M$ distinct molecules.
In this case, it is reasonable to study the amount of bits that can be stored relative to the number of potentially distinct molecules. 
Towards this goal we define the short-molecule rate $\tilde R$ as 
\begin{align}
\tilde R \defeq \frac{\log|\C|}{\M^\beta \len}.
\end{align}

\begin{proposition}
\label{thm:shortseqs}
Suppose that each molecule is drawn $N_i \sim Q$ times, with expectation $\Ex{N_i}>0$, and that  $\beta < 1$. Then, any achievable short-molecule rate satisfies $\tilde R \leq 1/\beta - 1$.
\end{proposition}

The proof, provided in the appendix, is based on the genie-aided and counting-based argument used in Section~\ref{sec:converseCounting}. 
The proposition guarantees that the (true) rate $R$ tends to zero at least as $1/M^{1-\beta}$.
While at first sight, it might seem surprising that there is no dependency on $\Emolseen$, this is reasonable, since in the regime of $\beta < 1$, no more than $\M^\beta$ distinct molecules exist. 
Thus, we see each fragment about 
$\EX{\N}/\M^\beta = 
\EX{N_i} \M/\M^\beta = \EX{N_i} \M^{1-\beta}$ many times, which tends to infinity, regardless of $Q$.

We point out that index-based coding schemes cannot achieve the scaling $R = \Theta(M^\beta L)$ suggested by the proposition.
To see this, suppose we encode the sequences by using $\len - 1$ bits for the index and only one bit for the information, and repeat each such segment $\M / (2^{\len - 1}) = 2 \M^{1-\beta}$ many times. 
We see each segment at least once with probability one as $\M \to \infty$. Thus we reliably store $2^{\len - 1} = \M^\beta /2 $ bits. 
Simple variations of this scheme (where we change the number of bits allocated to the index) can be similarly shown to only encode $\Theta(M^\beta)$ bits reliably.
Hence, for the regime $\beta \leq 1$, our upper bound to the number of bits that can be reliably stored is $\Theta(M^\beta L)$, while our lower bound is $\Theta(M^\beta)$, and it is an open question what the correct scaling is.



\subsection{Outlook}

In this paper we took steps towards the understanding of the fundamental limits of DNA-based storage systems. 
We proposed a simple model capturing the fact that molecules are stored in an unordered fashion, are short, and are corrupted by individual base errors. 
Our results show that a simple index-based coding scheme is asymptotically optimal for a large set of parameter choices.

While the model captures (moderate) substitution errors which are the prevalent error source on a nucleotide level of current DNA storage systems, the current generation of systems relies on \emph{low-error} synthesis and sequencing technologies that are relatively expensive and limited in speed.
A key idea towards developing the next-generation of DNA storage systems is to employ \emph{high-error}, but cheaper and faster synthesis and sequencing technologies such as light-directed maskless synthesis of DNA and nanopore sequencing. 
Such systems induce a significant amount of insertion and deletion errors. Thus, and important area of further investigation is to understand the capacity of channels which introduce deletions and insertions as well.

\subsection*{Acknowledgements}
IS and RH thank Kannan Ramchandran and David Tse for helpful discussions in the early stages of this work. RH would like to thank Robert Grass for helpful discussions on modeling aspects of DNA storage channels.

\printbibliography

\newpage

\appendix


\section{Proof of Lemma~\ref{lem:comb} }
Notice that vectors  $x \in \Z_+^{a}$ with $\|x \|_1 = b$ are in one-to-one correspondence with 
binary strings containing $(a-1)$ $0$s and $b$ $1$s.
For $x = (x_1,\ldots,x_a)$,
the corresponding string is 
\al{
\underbrace{1 \, \ldots \, 1}_{x_1} \, 0 \, \underbrace{1 \, \ldots \, 1}_{x_2} \, 0 \, \ldots \, 0 
\underbrace{1 \, \ldots \, 1}_{x_a}.
}
It is clear that such a string has $(a-1)$ $0$s and $b$ $1$s,
and that distinct strings with $(a-1)$ $0$s and $b$ $1$s correspond to distinct vectors $x$.
The number of distinct strings of this form is
\aln{
\frac{(a-1+b)!}{(a-1)!\, b!} = {a+b-1 \choose b}.
}
The upper bound in the statement of the lemma is a standard bound for binomial coefficients.


\section{Proof of Lemma~\ref{lem:conc} under a sampling-with-replacement model} 
\label{app:lemma2}

As it turns out, Lemma~\ref{lem:conc} can be proved under a sampling-with-replacement model.
Under this model, instead of sampling each molecule according to a probability distribution $Q$, $N$ sequences are sampled out of the pool of $M$ stored sequences.
Since there are multiple copies of each molecule in the pool due to PCR, we consider a sampling with replacement model.
By proving Lemma~\ref{lem:conc} in this setting,
one can establish a version of Theorem~\ref{thm:noisefree} for the sampling-with-replacement shuffling channel, as previously described in \cite{heckel_fundamental_2017}.

Consider the same genie-based argument described in Section~\ref{sec:converseCounting}.
In the sampling-with-replacement setting, the $\ell_1$ norm of the frequency vector $\vf$ at the output of the genie-aided channel is distributed as the number of distinct coupons obtained by drawing $\N = \cNM \M$ times with replacement from a set of $\M$ distinct coupons. 
Thus, Lemma~\ref{lem:conc} is an immediate consequence of the following stronger statement. 

\begin{lemma}
Let $Q$ be the number of distinct coupons obtained by drawing $\N = \cNM \M$ times with replacement from a set of $\M$ distinct coupons. 
We have that, for any $\delta > 0$,
\begin{align*}
\PR{Q \geq (1- e^{-\cNM} + \delta) \M}
\leq
\frac{1}{\M} \frac{2e^{2\cNM}}{2\left( \ln \left( \frac{e^{-\cNM}}{e^{-\cNM} - \delta} \right) - \frac{e^{\cNM}}{\M}\right)^2}.
\end{align*}
\end{lemma}

\begin{proof}
Since $\PR{Q \geq (1- e^{-\cNM} + \delta) \M}$ is a non-increasing function of $\delta$, we can assume that $\delta \in (0, e^{-\cNM}/2]$, as that simplifies the expressions.
Let $t_i$ be the number of draws to collect the $i$-th coupon after $(i-1)$ coupons have been collected, $i=0,\ldots,\M-1$, and consider the number of draws for obtaining $\alpha\M$ distinct coupons 
$
T \defeq \sum_{i=0}^{\alpha \M-1} t_i$
where
$\alpha \defeq 1 - e^{-\cNM} + \delta.
$
Due to
\[
\PR{Q \geq (1- e^{-\cNM} + \delta) \M}
=
\PR{Q \geq \alpha \M}
= 
\PR{T \leq \N}, 
\]
the lemma will follow by upper-bounding $\PR{T \leq \N}$ using Chebyshev's inequality. 
We first note that with 
$\EX{t_i} = 1/p_i, p_i \defeq \frac{\M-i}{\M}$ and
$\Var{t_i} = \frac{1-p_i}{p_i^2}$, 
we obtain 
\begin{align}
\EX{T} &
=
\sum_{i=0}^{\alpha \M-1} \EX{t_i}
= 
\M \sum_{i=0}^{\alpha \M-1} \frac{1}{\M-i}  \nonumber \\
&= 
\M(H_{\M} - H_{\M (1-\alpha) }) \nonumber \\
&\geq
\M( \ln \M - \ln (\M (1-\alpha)) ) - \frac{1}{2(1-\alpha)} \nonumber  \\
&\geq
-\M \ln(1-\alpha) - e^{\cNM}
= -\M \ln(e^{-\cNM} - \delta) - e^{\cNM} \nonumber \\
&= 
\M \cNM 
+ \M \underbrace{
\ln \left( \frac{e^{-\cNM}}{e^{-\cNM} - \delta} \right)
}_{\xi} - e^{\cNM}
= \N + \M \xi - e^{\cNM}. \nonumber
\end{align}
Here, $H_\M = \sum_{i=1}^\M \frac{1}{i}$ is the $\M$-th harmonic number, and the first inequality follows by the asymptotic expansion
\[
0 \leq H_n - \ln n - \gamma 
= 
\frac{1}{2n}  - \frac{1}{12 n^2} +  \frac{1}{120 n^4}  - \ldots 
\leq \frac{1}{2n},
\]
where $\gamma$ is the Euler-Mascheroni constant. 
The second inequality follows from $\frac{1}{1 - \alpha} \leq \frac{1}{e^{-\cNM} - e^{-\cNM}/2 } = 2 e^{\cNM}$. 
Moreover, the variance can be upper-bounded as
\begin{align}
\Var{T}
&=
\sum_{i=0}^{\alpha \M-1} \Var{t_i}
=
\sum_{i=0}^{\alpha \M-1} \frac{i\M}{(\M-i)^2} \nonumber \\
&
\leq 
\M \frac{\alpha}{2 (1-\alpha)^2}
\leq 
\M 2e^{2\cNM}.
\label{eq:boundvar}
\end{align}
Using the bound on the expectation and Chebyshev's inequality, we have for any $\beta>0$, that
\begin{align}
&\PR{ -T  + \N + \M \xi - e^{\cNM}  > \beta } \nonumber \\
&\hspace{2cm}\leq
\PR{ -T  + \EX{T}  > \beta }
\leq \frac{\Var{T}}{\beta^2}. \nonumber
\end{align}
Choosing $\beta = \M \xi - e^{\cNM}$ and using the upper bound on $\Var{T}$ given in~\eqref{eq:boundvar}, yields 
$
\PR{T \leq \N}
%
%
\leq 
\frac{1}{\M} \frac{2e^{2\cNM}}{\left(\xi - \frac{e^{\cNM}}{\M}\right)^2},
$
which concludes the proof. 
\end{proof}



\section{Proof of Lemma \ref{lem1}}


Let $Y^{\len}_1,\ldots, Y^{\len}_\N$ 
be the $\N$ strings observed at the output of the channel.
First we notice that, since $\N$ is a function of $Y^{\NL}$, we can write 
\al{
H\left(Y^{\NL}\right) & + H\left(S^N|X^{\ML},Y^{\NL}\right) \nonumber \\
&= H\left(Y^{\NL},N \right) + H\left(S^N|X^{\ML},Y^{\NL}, N \right) \nonumber \\
&= H\left(N \right) + H\left(Y^{\NL} | N \right) + H\left(S^N|X^{\ML},Y^{\NL}, N \right) \nonumber \\
&= H(N) + \sum_{n > 0} \Pr(N=n) \left[ H\left(Y^{\NL} | N = n \right) + H\left(S^N|X^{\ML},Y^{\NL}, N =n \right) \right]. \label{eq:lemproof1}
}
We will show that 
\al{
H\left(Y^{\NL} | N = n \right) + H\left(S^N|X^{\ML},Y^{\NL}, N =n \right) \leq nL + n \log \frac{M}{n} + o(ML),
}
which, when plugged back into \eqref{eq:lemproof1} implies that
\al{
H\left(Y^{\NL}\right) + H\left(S^N|X^{\ML},Y^{\NL}\right) 
& \leq \Ex[N] L + \Ex[N \log M/N] + o(ML) \nonumber \\
& \leq (1-q)ML + o(ML),
}
where we used the fact that $H(N) = o(ML)$, $\Ex[N] = (1-q)M$, and Jensen's inequality applied to the concave function $x \log (M/x)$.
This will establish the lemma.


In order to capture whether we are in the regime of Figure~\ref{fig:points}(a) or (b), we let 
$T$ be the largest subset of $[1:n]$
  so that, for any $i,j \in T$, 
$
d_H \left( Y_i^{\len}, Y_j^{\len} \right) \geq \alpha \len,
$
where $d_H$ is the Hamming distance and 
$\alpha > 2p$. 
We assume that in case of ties, an arbitrary tie-breaking rule is used to define $T$ (the actual choice will not be relevant for the proof).

Let $\Ex_n$ be the expectation conditioned on $N=n$; i.e.,  $\Ex_n[\cdot] = \Ex[\cdot | N=n]$.
We prove that, given the conditions in Lemma~\ref{lem1}, the following two bounds involving $\Ex_n|T|$ hold:
\begin{subequations}
\begin{align}
\text{(B1)\quad} & H\left(Y^{\NL}| N= n\right) \leq L \Ex_n|T|  
+ (n-\Ex_n|T|)\left( \log \Ex_n|T| + \len H(\alpha) \right) + o(\ML),
 \label{eq:B1}
 \\
\text{(B2)\quad} & H\left(S^N|X^{\ML},Y^{\NL}, N = n\right) \leq 
 n \log M
- \Ex_n|T| \log \Ex_n|T| + o(ML). 
 \label{eq:B2}
\end{align}
\end{subequations}
For large $\Ex_n|T|$, we are typically in the regime of Figure~\ref{fig:points}(a), while Figure~\ref{fig:points}(b) corresponds to the case where $\Ex_n|T|$ is small.
The bounds above capture the tension between the terms $H\left(Y^{\NL}|N=n\right)$ and $H\left(S^N|X^{\ML},Y^{\NL},N=n\right)$ because (B2) is decreasing in $\Ex_n|T|$, while (B1) is increasing in $\Ex_n|T|$ (provided that $\beta(1-H(\alpha)) \geq 1$).
Combining (B1) and (B2),
\al{
& H\left(Y^{\NL}|N=n\right) + H\left(S^N|X^{\ML},Y^{\NL},N=n\right) \nonumber \\
& \quad \leq  \len \Ex_n|T| + (n-\Ex_n|T|)\left( \log \Ex_n|T| + \len H(\alpha) \right) \nonumber \\
& \quad \quad \quad + n \log M - \Ex_n|T| \log \Ex_n|T| + o(\ML) \nonumber \\ 
& \quad = \Ex_n|T| L (1- H(\alpha))  + n \log \Ex_n|T|- 2 \Ex_n|T| \log \Ex_n|T| \nonumber \\
& \quad \quad \quad  + n \len H(\alpha) + n \log M + o(\ML).
\label{eq:boundET} 
}
Replacing $\Ex_n|T|$ with $x$ and ignoring the terms in this upper bound that do not involve $x$, we have the expression
\aln{
f(x) \defi  \gamma x  \log \M  + n \log x - 2 x\log x,
}
where we define $\gamma = \beta (1-H(\alpha))$.
For $x > 0$, we have 
\aln{
f'(x) & = \frac{1}{\ln(2)}\left( \gamma \ln \M + \frac{n}{x} - 2\ln x - 2 \right) \\
& > \frac{1}{\ln(2)}\left( \gamma \ln \M - 2\ln x - 2 \right) \\
& = \frac{2}{\ln(2)}\left( \ln \frac{\M^{\gamma/2}}{x} - 1 \right).
}
Hence $f'(x) > 0$ if
\al{
x < e^{-1} M^{\gamma/2}. \label{eq:maxx}
}
We see that, as long as $\gamma > 2$, the right-hand side of (\ref{eq:maxx}) is greater than $M$ for $M$ large enough.
This means that $f(x)$ is increasing for $1 \leq x \leq M$.
Since $\Ex_n|T| \leq n \leq M$, $f$ must attain its maximum at $f(n)$.
Therefore, (\ref{eq:boundET}) can be upper-bounded by setting $x = \Ex_n|T| = n$, which yields
\aln{
H\left(Y^{\NL}| N = n\right) & + H\left(S^N |X^{\ML},Y^{\NL},N=n\right) \leq nL + n \log \frac{M}{n}+ o(\ML).
}
Notice that this holds if, for some $\alpha > 2p$,
\aln{
\gamma = \beta(1-H(\alpha)) > 2 \Leftrightarrow 1 - H(\alpha) - 2/\beta > 0.
}
From the continuity of $H(\cdot)$, such $\alpha$ can be found if \eqref{eq:condition} holds, proving the lemma.
It remains to prove (B1) and (B2).

\subsection{Proof of (B1)}
Since $T$ is a deterministic function of $Y^{\NL}$ and can take at most $2^n$ values, 
\al{
H\left(Y^{\NL}| N = n\right) 
&= H\left(Y^{\NL}, T | N = n \right)  \nonumber \\
&=  H\left( T | N=n \right) +  H\left(Y^{\NL} | T, N=n \right) \nonumber\\
&\leq n + \sum_{t \subseteq [1:n]} \Pr\left( T = t | N = n \right) H\left(Y^{\NL} | T = t, N=n \right). \label{eq:bound3} 
}
Next we notice that, for a given $t$, we can write
\al{
H\left(Y^{\NL} | T = t, N=n \right) & = H\left( [Y^{\len}_i : i \in t] | T = t, N=n \right) \nonumber \\
& \hspace{-8mm} + H\left( [Y^{\len}_i : i \not\in t] | T = t, N=n, [Y^{\len}_i : i \in t] \right). \label{eq:bound4} 
}
The first term in  (\ref{eq:bound4}) is trivially bounded as 
\aln{
H\left( [Y^{\len}_i : i \in t] | T = t, N=n \right) \leq |t| L.
}
Each of the remaining length-$\len$ strings $Y^{\len}_i$ with $i \notin t$ must be within a distance $\alpha L$ from one of the strings in $[Y^{\len}_i : i \in t]$, from the definition of $T$.
Hence, conditioned on $[Y^{\len}_i : i \in t]$, each of them can only take at most $|t| |B(\alpha L)|$ values, where $B(\alpha L)$ is a Hamming ball of radius $\alpha \len$.
Since $|B(\alpha \len)| \leq 2^{\len H(\alpha)}$ for $\alpha < 1/2$, we bound the second term in (\ref{eq:bound4}) as
\aln{
& H\left( [Y^{\len}_i : i \not\in t] | T = t, N=n, [Y^{\len}_i : i \in t] \right) 
\leq (n-|t|) \left( \log |t| + \len H(\alpha) \right).
}
Using these bounds back in (\ref{eq:bound3}), we obtain
\al{
\hspace{-2mm} H\left(Y^{\NL} | N=n\right)  \nonumber 
& \leq n + \Ex_n  \left[  L |T| + (n-|T|) \left( \log |T| + \len H(\alpha) \right) | N = n \right] + o(\ML) \nonumber \\ & \leq L \Ex_n|T| + (n-\Ex_n|T|)\left( \log \Ex_n|T| + \len H(\alpha) \right) + o(ML),
\label{eq:bound5}
}
where we used the fact that $(n-x) \log x$ is a concave function of $x$ and Jensen's inequality. 





\subsection{Proof of (B2)}
Since $T$ is a deterministic function of $Y^{\NL}$,
\al{
& H\left(S^N|X^{\ML},Y^{\NL}, N = n\right) = H\left(S^N|X^{\ML},Y^{\NL}, T, N = n\right) \nonumber \\
& =  \sum_{t \subseteq [1:n]} \Pr\left( T = t | N = n\right) H\left(S^N|X^{\ML},Y^{\NL}, T = t, N = n\right) \nonumber \\
& \leq  \sum_{t \subseteq [1:n]} \Pr\left( T = t | N = n\right) \sum_{i=1}^n H\left(S(i)|X^{\ML},Y^{\NL}, T = t , N = n\right). \label{eq:bound6} 
}
Next we notice that the probability that $\delta \len$ or more errors occur in a single length-$\len$ string, for $\delta > p$, is at most $2^{-\len D(\delta \| p)}$ by the Chernoff bound (where $D(\cdot \| \cdot)$ is the binary KL divergence).
If we let $\E_i$ be the event that $d_H\left( X_{S(i)}^\len, Y_i^\len \right) \geq \delta L$, then we have
\aln{
\Pr(\E_i) \leq 2^{-\len D(\delta \| p)} = \M^{-\beta D(\delta \| p)}.
}
The conditional entropy term in (\ref{eq:bound6}) is upper bounded by
%
%
\al{
& H\left(S(i), \one_{\E_i}|X^{\ML},Y^{\NL}, T = t, N = n\right) \nonumber \\
& \quad  \leq H(\one_{\E_i}| T=t,N = n) + \Pr(\E_i | T = t, N = n) H\left(S(i)|X^{\ML},Y^{\NL}, T = t, N = n, \E_i \right) \nonumber \\
& \quad \quad \quad + \Pr(\bar\E_i | T = t, N = n) H\left(S(i)|X^{\ML},Y^{\NL}, T = t, N = n, \bar\E_i \right) \nonumber \\
& \quad \leq 1 + \Pr(\E_i | T = t, N = n) \log \M \nonumber \\
& \quad \quad \quad + H\left(S(i)|X^{\ML},Y^{\NL}, T = t, N = n, \bar\E_i \right)
\label{eq:bound7}
}
The final step is to bound the conditional entropy term in (\ref{eq:bound7}), for the case where $i \in t$.
Set $\delta = \alpha/2$. 
Conditioned on $\bar\E_i$, $d_H\left( X_{S(i)}^\len, Y_{i}^\len \right) < \alpha \len/2$.
Moreover, conditioned on $T=t$, for any $j \in t - \{i\}$, $d_H\left( Y_i^\len, Y_j^\len \right) \geq \alpha \len$.
For $i \in t$, we define the set 
%
%
%
%
%
%
%
%
\aln{
A_i = \{ j \st Y^{\len}_i \text{ is the closest output string in $t$ to }X^{\len}_j \}.
}
Notice that $A_i$, $i\in t$, forms a partition of $[1:M]$.
We claim that, if $i \in t$, $S(i)$ must be in $A_i$.
To see this notice that, for any $k \in t$, $k \ne i$, we have
\aln{
\alpha \len & \leq d_H \left( Y_i^{\len}, Y_k^{\len} \right) \\
& \leq d_H \left( X_{S(i)}^{\len}, Y_i^{\len} \right) + d_H \left( X_{S(i)}^{\len}, Y_k^{\len} \right)  \\
& < \alpha \len/2 + d_H \left( X_{S(i)}^{\len}, Y_k^{\len} \right),
}
implying that $d_H \left( X_{S(i)}^{\len}, Y_k^{\len} \right) > \alpha \len/2 \geq d_H \left( X_{S(i)}^{\len}, Y_i^{\len} \right)$, and thus $S(i) \in A_i$.
%
Therefore, $S(i)$ for each output string $Y^{\len}_i$ with $i \in t$, can take at most $|A_i|$ values.
Hence we have
\al{
& \sum_{i=1}^n H\left(S(i)|X^{\ML},Y^{\NL}, T = t, N = n, \bar\E_i \right) \nonumber \\
& \quad \quad \quad \leq \sum_{i \not\in t} \log M + \sum_{i \in t} \log |A_i| \nonumber \\
& \quad \quad \quad = (n-|t|) \log M + \sum_{i \in t} \log |A_i| \nonumber \\
& \quad \quad \quad \leq (n-|t|) \log M +|t| \log (M/|t|) \nonumber \\
& \quad \quad \quad = n \log M - |t| \log |t|, \label{eq:boundAi}
}
where the last inequality follows because $\sum_{i \in t} |A_i| = M$, and the sum is maximized by $|A_i| = \M/|t|$.
Combining (\ref{eq:bound6}), (\ref{eq:bound7}), and (\ref{eq:boundAi}), we obtain
\al{ \rescnt
& H\left(S^N|X^{\ML},Y^{\NL}, N=n\right)  \nonumber \\
& \quad  \leq  \sum_{i=1}^n \sum_{t \subseteq [1:n]} \Pr\left( T = t | N=n \right)\left[ 1 + \Pr(\E_i | T = t, N=n) \log \M  \right] \nonumber \\
& \quad  \;\; +  \sum_{t \subseteq [1:n]} \Pr\left( T = t | N=n\right) \sum_{i=1}^n H\left(S(i)|X^{\ML},Y^{\NL}, T = t, N=n, \bar \E_i \right) \nonumber \\
& \quad  =  n + \log M\sum_{i=1}^n \Pr(\E_i | N=n) + n \log M - \Ex_n\left[|T| \log |T|\right] \nonumber \\
& \quad  \leqnum  n + M \log M \Pr(\E_i) + n \log M - \Ex_n\left[|T| \log |T|\right] \nonumber \\
& \quad  \leqnum  n + M^{-\beta D(\delta \| p)} M \log M + n \log M - \Ex_n |T| \log \Ex_n |T| \nonumber
}
%
%
%
where, in $(i)$ we used the fact that $\E_i$ is independent of $N=n$ and $n \leq M$, and in $(ii)$ we used Jensen's inequality.
Since $ \M^{-\beta D(\delta || p)} \to 0$ as $\M \to \infty$,
$\M^{-\beta D(\delta || p)} \M \log \M = o(\ML)$, concluding the proof.

\section{Proof of Proposition~\ref{thm:shortseqs} }

We use a similar genie-aided and counting-based proof as in Section~\ref{sec:converseCounting}. 
The only difference is on how the number of frequency vectors is bounded. 
As before, the frequency vector on the output of the genie-aided channel satisfies, for any $\delta>0$, 
$\norm[1]{\vf}  \leq \M (\Emolseen + \delta)$. 
We next upper bound the number of different frequency vectors
$\vf \in \mathbb Z_+^{\M^\beta}$ with $\norm[1]{\vf} =  \M (\Emolseen + \delta)$. 
By Lemma~\ref{lem:comb}, the number of different frequency vectors we see at the output is upper bounded by
\aln{
\T[\M^\beta, \M(\Emolseen + \delta)] 
& = {\M^\beta+\M(\Emolseen + \delta)-1 \choose \M(\Emolseen + \delta) }
= {\M^\beta+\M(\Emolseen + \delta)-1 \choose  \M^\beta-1 }
\\
& < 
\left( \frac{e( \M^\beta+\M(\Emolseen + \delta)) }{ \M^\beta} \right)^{\M^\beta},
}
where the second equality follows from ${n \choose k} = {n \choose n-k}$.
Taking the logarithm we get
\[
\log \T[\M^\beta, \M(\Emolseen + \delta)] 
\leq
\M^\beta ( (1-\beta) \log \M + \log(1-q_0+\delta) + 1).
\]
Dividing by $\M^\beta \len = \M^\beta \beta \log(\M)$ and letting $\M \to \infty$ gives
\[
\tilde R \leq (1-\beta)/\beta,
\]
as desired.

\end{document}